\documentclass{conm-p-l}
\newtheorem{theorem}{Theorem}
\newtheorem{lemma}{Lemma}
\newtheorem{assumption}{Assumption}

\theoremstyle{remark}

\numberwithin{equation}{section}

\newcommand{\Cx}{{\mathbb C}}

\newcommand{\bN}{{\mathbb N}}

\newcommand{\Ir}{{\mathbb Z}}
\newcommand{\A}{{\mathcal A}}
\newcommand{\cA}{{\mathcal A}}
\newcommand{\cB}{{\mathcal B}}

\newcommand{\cL}{{\mathcal L}}
\newcommand{\cK}{{\mathcal K}}

\def\idty{{\mathchoice {\mathrm{1\mskip-4mu l}} {\mathrm{1\mskip-4mu l}} %
{\mathrm{1\mskip-4.5mu l}} {\mathrm{1\mskip-5mu l}}}}
\newcommand{\id}{\mathop{\rm id}}

\newcommand{\boxendproof}{\hspace*{\fill}{{$\Box$}} \vspace{10pt}}

\newcommand{\be}{\begin{equation}}
\newcommand{\ee}{\end{equation}}
\newcommand{\bea}{\begin{eqnarray}}
\newcommand{\eea}{\end{eqnarray}}
\newcommand{\beann}{\begin{eqnarray*}}
\newcommand{\eeann}{\end{eqnarray*}}
\newcommand{\eq}[1]{(\ref{#1})}

\usepackage{color}

\begin{document}

\title[Irreversible Dynamics]{Lieb-Robinson Bounds and Existence of the Thermodynamic Limit
for a Class of  Irreversible Quantum Dynamics}

\author[B. Nachtergaele]{Bruno Nachtergaele}
\address{Department of Mathematics, University of California at Davis,
Davis, California 95616}
\email{bxn@math.ucdavis.edu}

\author[A. Vershynina]{Anna Vershynina}
\address{Department of Mathematics, University of California at Davis,
Davis, California 95616}
\email{aver@math.ucdavis.edu}

\author[V.A. Zagrebnov]{Valentin A. Zagrebnov}
\address{
Universit\'{e} de la M\'editerran\'ee (Aix-Marseille II)\\
Centre de Physique Th\'eorique- UMR 6207 CNRS, Luminy - Case 907\\
13288 Marseille, Cedex 09, France}
\email{zagrebnov@cpt.univ-mrs.fr}

%    Fix General info

%Fix keywords

\dedicatory{Dedicated to Robert A. Minlos at the occasion of his 80th birthday}

\begin{abstract}
We prove Lieb-Robinson bounds and the existence of the thermodynamic limit
for a general class of irreversible dynamics for quantum lattice systems with
time-dependent generators that satisfy a suitable decay condition in space.
\end{abstract}

\maketitle
\footnotetext[1]{Copyright \copyright\ 2011 by the authors. This
paper may be reproduced, in its entirety, for non-commercial
purposes.}

\section{Introduction}\label{sec:intro}

For a quantum many-body Hamiltonian describing bulk matter we expect
that the Heisenberg dynamics converges in the thermodynamic limit
to a well-defined one-parameter flow of transformations on the observable algebra.
Early results of this kind were obtained for quantum spin systems
\cite{streater:1968,robinson:1968,ruelle:1969}, which were followed by
generalizations that included examples of irreversible dynamics
described
by a semigroup of completely positive unit preserving maps \cite{gorini:1976,spohn:1980}.
See, e.g., \cite{davies:1977,fannes:1978,bratteli:1980,unnerstall:1990,matsui:1993}
for results on the thermodynamic limit of a number of examples of semigroups
of completely positive maps.
In this work we study a general class of irreversible dynamics for quantum lattice
systems with generators that are sums of bounded terms that may depend on time
and that satisfy a suitable decay condition in space.

Following the argument of \cite{robinson:1976} propagation bounds of Lieb-Robinson
type \cite{lieb:1972} have recently been used to prove a number of new results on
the existence of the thermodynamic limit  \cite{nachtergaele:2006,amour:2009,
nachtergaele:2010, bachmann:2011}. These recent developments were made possible
by extensions and improvements of the Lieb-Robinson bounds themselves
\cite{hastings:2004,nachtergaele:2006a,hastings:2006,nachtergaele:2009a,nachtergaele:2009b}.

Lieb-Robinson type bounds for irreversible dynamics were, to our knowledge, first
considered in \cite{hastings:2004b} in the classical context and in \cite{poulin:2010}
for a class of quantum lattice systems with finite-range interactions. Here, we will extend
those results by proving a Lieb-Robinson bound for lattice models with a dynamics generated
by both Hamiltonian and dissipative interactions with suitably fast decay in space and
that may depend on time. See Assumption \ref{assumption1} for the precise conditions.
Then, we use our result to prove the existence of the thermodynamic
limit of the dynamics in the sense of a strongly continuous one-parameter flow of completely
positive unit preserving maps.

Our results are applicable to a wide range of model systems in statistical mechanics, quantum
optics, and quantum information and computation. In each of those areas, it is often necessary
to incorporate dissipative and time-dependent terms in the generator of the dynamics.
Fortunately, there is a large number of interesting systems defined on a lattice, which
so far is the only setting accessible by our methods to prove Lieb-Robinson bounds.
It is probably not a coincidence that proofs of the existence of the thermodynamic limit of
the dynamics have so far also been mostly restricted to lattice systems. Here, `lattice'
has to be interpreted loosely to mean a discrete set of points that are typically thought of as
distributed in space. In the case of the positions of atoms in a crystal, these positions
can indeed be described by a lattice, but all one needs is the structure of a metric graph
satisfying some regularity conditions. The detailed setup is given in Section
\ref{sec:setup}.

The existence of the thermodynamic limit is important as a fundamental property
of any model meant to describe properties of bulk matter. In particular, such properties
should be essentially independent of the size of the system which, of course, in any
experimental setup will be finite. In the past five years, Lieb-Robinson bounds have been
used to prove a variety of interesting results about condensed matter systems. 
See \cite{nachtergaele:2010b} for a brief overview of the applications of Lieb-Robinson bounds.

The paper is organized as follows. First, we describe the general setup necessary to
state our main results, which we do in Section \ref{sec:setup}. In that section we also state
the three main theorems we prove in this paper. Theorem \ref{thm:finite} states
that solution of the differential equation (master equation)
defined by finite volume generators we consider is a well-defined quantum dynamics,
i.e., a continuous family of completely positive unit preserving maps on the algebra of
observables. The proof of this theorem is obtained by standard methods, but for
completeness we included it here in Section \ref{sec:finite}. Theorem \ref{thm:lrbounds}
is the Lieb-Robinson bound, i.e., the propagation estimate for irreversible dynamics.
Again the theorem is stated in Section \ref{sec:setup} and then proved
in Section \ref{sec:lrbounds}. Theorem  \ref{thm:thermodynamiclimit}, the existence of 
the thermodynamic limit, is proved in Section \ref{sec:thermodynamiclimit}.

\section{Setup and main results}\label{sec:setup}

We consider quantum systems consisting of components associated with the
vertices $x\in\Gamma$, where $\Gamma$ is a countable set equipped with a metric $d$.
We assume that there exists a non-increasing function $F: [0, \infty)\rightarrow (0, \infty)$ 
such that:\\
i) $F$ is uniformly integrable over $\Gamma$, i.e.,
\begin{equation*}
\|F\|:= \sup_{x\in\Gamma}\sum_{y\in\Gamma} F(d(x,y)) < \infty,
\end{equation*}
and\\
ii) $F$ satisfies
\begin{equation*}
C:= \sup_{x,y\in \Gamma}\sum_{z\in\Gamma}\frac{F(d(x,z))F(d(y,z))}{F(d(x,y))}  < \infty.
\end{equation*}

Having such a set $\Gamma$ and a function $F$ that satisfies i) and ii), we can define
for any $\mu>0$ the function
\begin{equation*}
F_\mu(x)=e^{-\mu x}F(x),
\end{equation*}
which then also satisfies i) and ii) with $\|F_\mu\|\leq \|F\|$ and $C_\mu\leq C$.

The Hilbert space of states of the subsystem at $x\in\Gamma$ is denoted by $\mathcal{H}_x$.
For any finite subset $\Lambda\subset\Gamma$  the Hilbert associated with $\Lambda$ is
$$
\mathcal{H}_\Lambda=\bigotimes_{x\in \Lambda}\mathcal{H}_x.
$$
The algebra of observables supported in $\Lambda$ is defined by
$$
\mathcal{A}_\Lambda=\bigotimes_{x\in \Lambda}\mathcal{B}(\mathcal{H}_x),
$$
where $\mathcal{B}(\mathcal{H}_x)$  is the set of bounded linear operators on $\mathcal{H}_x$.
If $\Lambda_1\subset \Lambda_2$, then we may identify $\mathcal{A}_{\Lambda_1}$ in
a natural way with the subalgebra
$\mathcal{A}_{\Lambda_1}\otimes \idty_{\Lambda_2\setminus\Lambda_1}$ of
$\mathcal{A}_{\Lambda_2}$, and simply write
$\mathcal{A}_{\Lambda_1}\subset\mathcal{A}_{\Lambda_2}$.
The algebra of local observables is then defined as
$$
\mathcal{A}_{\Gamma}^{\rm loc}= \bigcup_{\Lambda\subset\Gamma}\mathcal{A}_\Lambda.
$$
The $C^*$-algebra of quasi-local observables $\mathcal{A}_{\Gamma}$ is the norm
completion of $\mathcal{A}_{\Gamma}^{\rm loc}$. See \cite{bratteli:1987,bratteli:1997}
for more details about this mathematical framework.

The support of the observable $A\in\mathcal{A}_\Lambda$ is the minimal set
$X\subset\Lambda$ for which  $A=A^\prime\otimes \idty_{\Lambda\setminus X}$
for some $A^\prime\in\mathcal{A}_X$.

The generator of the dynamics is defined for each finite volume $\Lambda\subset\Gamma$,
and, in general, contains both Hamiltonian interactions and dissipative terms, which we
allow to be time-dependent. The Hamiltonian terms are described by an interaction $\Phi(t, \cdot)$
which, for all $t\in\mathbb{R}$, is a map from a set of subsets of $\Gamma$ to
$\mathcal{A}_\Gamma$, such that for each finite set
$X\subset\Gamma$, $\Phi(t,X)\in\mathcal{A}_X$ and $\Phi(t,X)^*=\Phi(t,X)$.
The dissipative part is described by terms of Lindblad form determined, for each finite
$X\subset\Gamma$,  by a set of operators $L_a(t, X)\in\mathcal{A}_X$, $a=1,\ldots, N(X)$.
We can allow for the case $N(X)=\infty$, if we impose a suitable convergence condition
on the resulting series for the generator. Then, for any finite set $\Lambda\subset\Gamma$
and time $t\in\mathbb{R}$ we define the family of bounded linear maps
$\mathcal{L}_\Lambda: \mathcal{A}_\Lambda \rightarrow \mathcal{A}_\Lambda,$ i.e.
$\mathcal{L}_\Lambda\in\mathcal{B}(\mathcal{A}_\Lambda,\mathcal{A}_\Lambda)$, as follows:
for all $A\in\mathcal{A}_\Lambda$,
\begin{align}
\Psi_Z(t)(A)&= i[\Phi(t,Z), A]\label{psiz}\\
&\quad+\sum_{a=1}^{N(Z)} [L^*_a(t,Z)AL_{a}(t,Z)-\frac{1}{2}\{L_{a}(t,Z)^*L_{a}(t,Z), A\} ]\nonumber\\
\mathcal{L}_\Lambda(t)(A)&=\sum_{Z\subset\Lambda}  \Psi_Z(t)(A),\label{Def_L}
\end{align}
where $\{A, B\}=AB+BA$, is the anticommutator of $A$ and $B$. The operators $\Psi_Z(t)$
can be regarded as bounded linear transformations on $\cA_X$, for any $X\subset\Lambda$
that contains $Z$, which are then of the form $\Psi_Z(t)\otimes \id_{\cA_{X\setminus Z}}$.
The norm of these maps, in general, depends on $X$, but they are uniformly
bounded as follows:
$$
\Vert \Psi_Z(t)\Vert\leq 2\Vert \Phi(t,Z)\Vert + 2 \sum_{a=1}^{N(Z)} \Vert L_a(t,Z)\Vert^2
$$
If $N(z)=\infty$, we can insure uniform boundedness by assuming that the sums
$\sum_{a=1}^{\infty} \Vert L_a(t,Z)\Vert^2$ converge. It is more general and more natural,
however, to assume that  the maps $\Psi_Z(t)$, defined on $\cA_Z$ are
completely bounded. By definition, $\Psi\in\cB(\cA_Z)$ is called {\em
completely bounded} if for all $n\geq 1$, the linear maps $\Psi\otimes\id_{M_n}$,
defined on $\A_Z\otimes M_n$, where $M_n = \cB(\Cx^n)$ are the $n\times n$ complex
matrices, are bounded with uniformly bounded norm. This means that we can define 
the cb-norm of $\Psi$ by
$$
\Vert \Psi\Vert_{\rm cb}=\sup_{n\geq 1} \Vert \Psi\otimes {\rm id}_{M_n}\Vert<\infty\,
$$
In particular, this definition implies that the cb-norm of $\Psi_Z(t)$, which can be considered
as a linear map defined on $\cA_\Lambda$ for all $\Lambda\subset\Gamma$ such that  
$Z\subset \Lambda$, is {\em independent} of $\Lambda$.
See \cite{effros:2000,paulsen:2002} for more information on completely bounded maps.
Assuming that $\Vert \Psi_Z(t)\Vert_{\rm cb}$ is finite is more general than assuming
that the series $\sum_{a=1}^{\infty} \Vert L_a(t,Z)\Vert^2$ converges which, however,
is a useful sufficient condition for it. In particular, there are situations where the sum in
\eq{psiz} only converges in the strong operator topology but nevertheless yields
a well-defined limit with finite cb-norm.

\begin{assumption}\label{assumption1}
Given $(\Gamma, d)$ and $F$ as described at the beginning of this section, the following hypotheses hold:

\begin{enumerate}
\item
For all finite $\Lambda\subset\Gamma$, $\mathcal{L}_\Lambda(t)$ is norm-continuous in $t$,
and hence uniformly continuous on compact intervals.
\item
There exists $\mu>0$ such that for every $t\in\mathbb{R}$
\begin{equation}\label{Decay_L}
\|\Psi\|_{t,\mu}:=\sup_{s\in[0,t]}\sup_{x,y\in\Lambda}
\sum_{Z\ni x,y}\frac{\|\Psi_Z(s)\|_{\rm cb}}{F_\mu(d(x,y))}<\infty.
\end{equation}
where $\Vert\cdot\Vert_{\rm cb}$ denotes the cb-norm of completely bounded
maps \cite{paulsen:2002}.
\end{enumerate}
\end{assumption}

Note that 
$$
\|\mathcal{L}_\Lambda(t)\|\leq \sum_{Z\subset\Lambda}\|\Psi_Z(t)\|\leq\sum_{x,y\in\Lambda}
\sum_{Z\ni x,y}\|\Psi_Z(t)\|_{\rm cb}\leq \|\Psi\|_{t,\mu}|\Lambda|\|F\|.
$$
We define
\begin{equation}\label{Mt}
M_t=\|\Psi\|_{t,\mu}|\Lambda|\|F\|\, .
\end{equation}
Then by (\ref{Decay_L}) one gets $M_s\leq M_t$ for $s<t$.

Fix $T>0$ and, for all $A\in\A_\Lambda$, let $A(t), t\in [0,T]$ be a solution of the initial
value problem
\begin{equation}\label{ode}
 \frac{d}{dt} A(t)=\cL_\Lambda(t)A(t),\quad A(0)=A.
\end{equation}
Since $\|\mathcal{L}_\Lambda(t)\|\leq M_T<\infty$, this solution exists and is unique by the
standard existence and uniqueness results for ordinary differential equations.
For $0\leq s\leq t\leq T$, define the family of maps
$\{\gamma_{t,s}^\Lambda\}_{0\leq s\leq t} \subset \cB(\cA_\Lambda, \cA_\Lambda)$ by $\gamma_{t,s}^\Lambda(A)=A(t)$, 
where $A(t)$ is the unique solution of \eq{ode} for $t\in [s,T]$ with initial condition $A(s)=A$. Then,
the {\em cocycle property}, $\gamma_{t,s}(A(s))=A(t)$, follows from the uniqueness of the
solution of \eq{ode}. Recall that a linear map $\gamma:\cA\to\cB$, where $\cA$ and $\cB$
are $C^*$-algebras is called {\em completely positive} if the maps $\gamma\otimes\id:
\cA\otimes M_n \to \cB\otimes M_n$ are positive for all $n\geq 1$. Here $M_n$ stands
for the $n\times n$ matrices with complex entries, and positive means that positive
elements (i.e., elements of the form $A^*A$)  are mapped into positive elements.
See, e.g., \cite{paulsen:2002} for a discussion of the basic properties of completely positive
maps. In particular, we shall use the property that every unit preserving (i.e. $\gamma(\idty_\cA)=\idty_\cB $)
completely positive map $\gamma$,  is a contraction: $\Vert \gamma(A)\Vert
\leq \Vert A\Vert$.

As a preliminary result we prove the following Theorem \ref{thm:finite} in Section \ref{sec:finite}.
It extends the well-known result for time-independent generators of Lindblad form
\cite{lindblad:1976} to the time-dependent case.

\begin{theorem}\label{thm:finite}
Let $\cA$ be a $C^*$-algebra, $T>0$, and for $t\in [0,T]$, let $\cL(t)$ be a norm-continuous
family of bounded linear operators on $\cA$.
If\newline
(i) $\mathcal{L}(t)(\idty)=0;$\newline
(ii) for all $A\in\cA$, $\mathcal{L}(t)(A^{*})=\mathcal{L}(t)(A)^{*} $;\newline
(iii) for all $A\in\cA$,   $\mathcal{L}(t)(A^* A)-\mathcal{L}(t)(A^*)A-A^*\mathcal{L}(t)(A)\geq 0$;\newline
then the maps $\gamma_{t,s}$, $0\leq s\leq t\leq T$, defined by equation \eq{ode}, are a
norm-continuous cocycle of unit preserving completely positive maps.
\end{theorem}

It is straightforward to check that the $\cL_\Lambda(t)$ defined in \eq{Def_L}
satisfy properties (i) and (ii). Property (iii), which is called {\em complete dissipativity},
follows immediately from the observation
$$
\mathcal{L}(t)(A^* A)-\mathcal{L}(t)(A^*)A-A^*\mathcal{L}(t)(A)
=\sum_{Z\subset\Lambda}\sum_{a=1}^{N(Z)} [A, L_a(t,Z)]^*[A, L_a(t,Z)]\, \geq 0 \, .
$$
Therefore,  using this result, we conclude that, under Assumption \ref{assumption1}, for all
finite $\Lambda\subset\Gamma$, the maps $\gamma_{t,s}^\Lambda$, $0\leq s\leq t$, form
a norm-continuous cocycle of completely positive and unit preserving maps.

Section \ref{sec:lrbounds} is devoted to proving a Lieb-Robinson bound for the
irreversible dynamics $\gamma_{t,s}^\Lambda$. For reversible dynamics given
by the one-parameter group of automorphisms $\tau_t$ describing the Heisenberg
dynamics generated by a Hamiltonian, Lieb-Robinson bounds take the following
form: there are constants $v, \mu>0$ such that  for $A\in A_X$ and $B\in \A_Y$,
\begin{equation}\label{lrrev}
\Vert [A,\tau_t(B)]\Vert \leq C(A,B) e^{-\mu(d(X,Y)-v|t|)}\,,
\end{equation}
where $d(X,Y)$ denotes the distance between $X$ and $Y$ and $C(A,B)$ is a
prefactor, which typically has the form $c\Vert A\Vert\, \Vert B\Vert \min(| X|,|Y|)$,
for a suitable norm $\Vert\cdot\Vert$ on the observables  $A$ and $B$,
and a suitable measure $|\cdot |$ on the size of the supports $X$ and $Y$.
Bounds of this form are sufficient to determine the approximate support of
the time-evolved observable $\tau_t(B)$. See, e.g., \cite[Lemma 3.1]{bachmann:2011}.

For irreversible dynamics, it turns out to be both natural and convenient to consider 
a slightly more general formulation. For $X\subset\Lambda$, let $\cB_X$ denote the 
subspace of $\cB(\cA_X)$ consisting of all completely bounded 
linear maps that vanish on $\idty$. See the discussion directly preceding Assumption
\ref{assumption1} for the definition of complete boundedness and the cb-norm 
$\Vert\cdot\Vert_{\rm cb}$.
It is important for us that  all operators of the form
$$
\mathcal{K}_X(B)= [A, B]+\sum_{a=1}^{N} [L^*_a B L_{a}-\frac{1}{2}\{L_{a}^*L_{a}, B\} ]\,,
$$
where $A, L_a\in \cA_X$, belong to $\cB_X$, with
$$
\Vert \cK_X\Vert_{\rm cb}\leq 2\Vert A\Vert +2\sum_{a=1}^N \Vert L_a\Vert^2.
$$
In particular, 
operators of the form $[A,\cdot]$ appearing in the standard Lieb-Robinson 
bound \eq{lrrev} are a special case of this general form. Then, we can regard 
$\cK_X$ as a linear transformation on $\cA_Z$, for all $Z$ such that $X\subset Z$,
by tensoring it with $\id_{\cA_{Z\setminus X}}$, and all these maps will be bounded
with norm less then $\Vert \cK_X\Vert_{\rm cb}$.

\begin{theorem}\label{thm:lrbounds}
Suppose Assumption \ref{assumption1} holds. Then  the maps $\gamma^\Lambda_{t,s}$
satisfy the following bound. For $X,Y\subset\Lambda$, and any operators $\cK\in\cB_X$ and
$B\in\cA_Y$ we have that
\begin{equation*}
\|\cK(\gamma_{t,s}^\Lambda(B))\|\leq \frac{\| \cK \|_{\rm cb}\,\|B\|}{C_\mu} e^{\|\Psi\|_{t,\mu}
C_\mu |t-s|}\sum_{x\in X}\sum_{y\in Y}F_\mu(d(x,y))\,.
\end{equation*}
\end{theorem}

Note that the bound in this theorem is \textit{uniform} in $\Lambda$. This is important
for the proof of existence of the thermodynamic limit of the dynamics, which
is the main application of Lieb-robinson bounds in the present paper. 

As a final 
comment about the use of the cb-norm in the definition of $\|\Psi\|_{t,\mu}$ (see \eq{Decay_L})
we would like to point out that volume-independent bounds for the operator
norm, such as $\Vert [\Psi, \cdot]\Vert \leq 2\Vert \Psi\Vert$, which appear in 
all previous Lieb-Robinson bounds, are always an upper bound for the norm
used here. This is also true for the case of reversible dynamics. The bound stated here 
will give a sharper result in some cases. It has been suggested that the
addition of dissipative terms to the generator of the dynamics would not 
increase the Lieb-Robinson velocity. For example, in \cite{schuch:2010} it is correctly
argued that the bounds derived in that paper remain valid without change
if one adds particle loss terms. While this is trivially true for arbitrary single-site terms, 
it is not clear that the same comparison would hold in general. We also need
to alert the reader that the bounds derived in \cite{schuch:2010}, while valid for 
lattice bosons with a finite number of particles, have a prefactor which depends 
on the particle number (at least linearly for the simplest observables,
and worse than linear for more general observables). 
In this sense the results of  \cite{schuch:2010} are not a true extension of
\cite{poulin:2010} to many-body boson systems as discussed, e.g., in the recent 
book \cite{verbeure:2011}.

Also note that the bound given in Theorem \ref{thm:lrbounds} can be further
improved by omitting in the definition of $\|\Psi\|_{t,\mu}$ all terms that act on a 
single site, i.e., which belong to $\cB_{\{x\}}$ for some $x\in\Gamma$, and also
all terms in $\cB_X$, where $X$ is the set for which the bound is derived. By the
argument in \cite{nachtergaele:2009a} we can even allow the single site terms
to be unbounded, as long as they lead to a well-defined single-site dynamics.

In this paper, we restrict ourselves to applying the Lieb-Robinson bound of Theorem 
\ref{thm:lrbounds} to proving the existence of the thermodynamic limit of 
a general family of irreversible dynamics. Other applications, such as approximate
factorization of invariant states, analogous to what is done for ground states
of reversible systems in \cite{hamza:2009}.

The setup for the analysis of the thermodynamic limit can be formulated as follows. 
Let $\Gamma$ be 
an infinite set such as, e.g., the hypercubic lattice $\Ir^\nu$. We prove the existence
of the thermodynamic limit for an increasing and exhausting sequence of finite
subsets $\Lambda_n\subset\Gamma$, $n\geq 1$, by showing that  for each $A\in \cA_X$,
$(\gamma^{\Lambda_n}_{t,s}(A))_{n\geq 1}$ is a Cauchy sequence in the norm of
$\cA_\Gamma$. To this end we have to suppose that Assumption \ref{assumption1} (2) holds 
\textit{uniformly} for all $\Lambda_n$, i.e., we can replace $\Lambda$ in \eq{Decay_L} by
$\Gamma$.

\begin{theorem}\label{thm:thermodynamiclimit}
Suppose that Assumption \ref{assumption1} holds and, in addition, that
\eq{Decay_L} holds for $\Lambda =\Gamma$. Then, there exists a strongly continuous 
cocycle of unit-preserving completely positive maps $\gamma_{t,s}^\Gamma$ on 
$\cA_\Gamma$ such that for all $0\leq s\leq t$, and any increasing exhausting sequence 
of finite subsets  $\Lambda_n\subset\Gamma$, we have
\begin{equation}
\lim_{n\rightarrow\infty}\|\gamma_{t,s}^{\Lambda_n}(A)-\gamma_{t,s}^\Gamma(A)\|=0,
\end{equation}
for all $A\in\cA_\Gamma$.
\end{theorem}

\section{Finite volume dynamics}\label{sec:finite}

Let $\cL(t)$, $t\geq 0$, denote a family of operators on a $C^*$-algebra $\cA$ satisfying
the assumptions of Theorem \ref{thm:finite} and for $0\leq s\leq t$ consider the maps
$\cA\ni A \mapsto \gamma_{t,s}(A)$ defined by the solutions of (\ref{ode}) with initial condition
$A$ at $t=s$. Without loss of generality we can assume $s=0$ in the proof of the theorem because,
if we denote $\tilde{\cL}(t)=\cL(t+s)$, then $\gamma_{t,s}=\tilde{\gamma}_{t-s,0}$, where $
\tilde{\gamma}_{t,0}$ is the maps determined by the generators $\tilde{\cL}(t)$.

The maps $\gamma_{t,s}$ satisfy the equation
\begin{equation}\label{sol_gamma}
\gamma_{t,s}=\id+\int_s^t \cL(\tau)\gamma_{\tau, s}d\tau.
\end{equation}
In our proof of the complete positivity of $\gamma_{t,0}$ we will use an expression
for $\gamma_{t,0}$ as the limit of an Euler product, i.e.. approximations $T_n(t)$ defined by
\begin{equation}\label{Tn}
T_n(t)= \prod_{k=n}^1\left(\id+\frac{t}{n}\mathcal{L}(\frac{kt}{n})\right)\, .
\end{equation}
The product is taken in the order so that the factor with $k=1$ is on the right.

 \begin{lemma}\label{lem:Euler_app}
Let $\cL(t)$, $t\geq 0$, denote a family of operators on a $C^*$-algebra $\cA$ satisfying
the assumptions of Theorem \ref{thm:finite}. Then, uniformly for all $t\in [0,T]$,
 \begin{equation*}
\lim_{n\to\infty} \Vert T_n(t) - \gamma_{t,0}\Vert = 0 \,,
 \end{equation*}
 where $T_n(t)$ is defined by \eq{Tn}.
 \end{lemma}

 \begin{proof}
 From the cocycle property established in Section \ref{sec:setup}, we have
 \begin{equation*}
\gamma_{t,0}=\prod_{k=n}^1 \gamma_{t\frac{k}{n}, t\frac{k-1}{n}}.
 \end{equation*}
 Now, consider the difference
 \begin{align*}
 T_n(t)-\gamma_{t,0}&=\prod_{k=n}^1\left(\id+\textstyle\frac{t}{n}\mathcal{L}(\textstyle\frac{kt}{n})\right)
 -\prod_{k=n}^1 \gamma_{t\frac{k}{n}, t\frac{k-1}{n}}\\
 &=\sum_{j=1}^n\left[\prod_{k=n}^{j+1}\left(\id+\textstyle\frac{t}{n}\mathcal{L}(t\textstyle\frac{k-1}{n})\right)\right]
 \left[\left(\id+\frac{t}{n}\mathcal{L}(t\textstyle\frac{j-1}{n})\right)-
 \gamma_{t\frac{j}{n}, t\frac{j-1}{n}}\right] \gamma_{t\frac{j-1}{n},0}.
 \end{align*}
  To estimate the norm of this difference we look at each factor separately.

Using the boundedness of $\cL(t)$ and the fact that $M_t$, defined in \eq{Mt},
is increasing in $t$, the norm of the first factor is bounded from above by
$$
\|\prod_{k=n}^{j+1}(\id+\frac{t}{n}\mathcal{L}(t\frac{k-1}{n}))\|\leq \prod_{k=n}^1(1+
\frac{t}{n}\|\mathcal{L}(t\frac{k-1}{n})\|)\leq (1+\frac{t}{n}M_t)^n.
$$

To bound the second factor notice that from (\ref{sol_gamma}) we obtain
\begin{equation*}
\|\gamma_{t,s}\|\leq 1+\int_s^t\|\cL(\tau)\|\|\gamma_{t,s}\|d\tau.
\end{equation*}
Then by Gronwall inequality \cite[Theorem 2.25]{hunter:2001} we have the following bound for 
the norm of the $\gamma_{t,s}$:
\begin{equation*}
\|\gamma_{t,s}\|\leq e^{\int_s^t\|\cL(\tau)\|d\tau}\leq e^{M_t(t-s)}.
\end{equation*}

Using again (\ref{sol_gamma}) we can rewrite the second factor as follows:
 \begin{align*}
 &\left(\id+\frac{t}{n}\mathcal{L}(t\frac{j-1}{n})\right)-\gamma_{t\frac{j}{n}, t\frac{j-1}{n}}=
 \frac{t}{n}\mathcal{L}(t\frac{j-1}{n})-(\gamma_{t\frac{j}{n}, t\frac{j-1}{n}}-\id)\\
& = \int_{t\frac{j-1}{n}}^{t\frac{j}{n}}
 \left (\cL(t\frac{j-1}{n})-\cL(s)\gamma_{s, t\frac{j-1}{n}}\right)ds\\
 &=\int_{t\frac{j-1}{n}}^{t\frac{j}{n}}\left[\left(\cL(t\frac{j-1}{n})-\cL(s)\right)-\cL(s)
 (\gamma_{s, t\frac{j-1}{n}}-\id)\right]ds\\
 &=\int_{t\frac{j-1}{n}}^{t\frac{j}{n}}\left(\cL(t\frac{j-1}{n})-\cL(s)\right) ds-\int_{t\frac{j-1}{n}}^{t\frac{j}{n}}\cL(s)
 \int_{t\frac{j-1}{n}}^s\cL(\tau)\gamma_{\tau, t\frac{j-1}{n}}d\tau ds.
 \end{align*}

Therefore, the second factor is bounded from above by
\begin{align*}
\|\left(\id+\frac{t}{n}\mathcal{L}(t\frac{j-1}{n})\right)-\gamma_{t\frac{j}{n}, t\frac{j-1}{n}}\|
&\leq\frac{t}{n}\epsilon_n+M_t^2\int_{t\frac{j-1}{n}}^{t\frac{j}{n}}\int_{t\frac{j-1}{n}}^s e^{(\tau-t\frac{j-1}{n})M_t}
d\tau ds\\
&\leq\frac{t}{n}\epsilon_n+M_t^2e^{\frac{t}{n}M_t}\int_{t\frac{j-1}{n}}^{t\frac{j}{n}} (s-t\frac{j-1}{n})ds\\&=\frac{t}{n}
\left(\epsilon_n+M_t^2 e^{\frac{t}{n}M_t}\frac{t}{2n}\right),
\end{align*}
where $\epsilon_n\to 0$ as $t/n\to 0$ due
to the uniform continuity of $\cL(t)$ on the interval $[0,t]$.

The third factor can be estimated in a similar way:
 \begin{align*}
 \|\gamma_{t\frac{j-1}{n},0}\|&=\prod_{k=j-1}^1\|\gamma_{t\frac{k}{n},t\frac{k-1}{n}}\|=\prod_{k=j-1}^1\|1+\frac{t}{n}
 \mathcal{L}(s_k(\frac{t}{n})) \|\\
 &\leq\prod_{k=j-1}^1\left(1+\frac{t}{n}\|\mathcal{L}(s_k(\frac{t}{n}))\|\right)\\
 &\leq (1+\frac{t}{n}M_t)^n.
 \end{align*}

Therefore, combining all these estimates we obtain
 \begin{align*}
\|T_n(t)-\gamma_{t,0}\|&\leq n (1+\frac{t}{n}M_t)^n \frac{t}{n}\left(\epsilon_n+M_t^2 e^{\frac{t}{n}M_t}\frac{t}{2n}\right)
 (1+\frac{t}{n}M_t)^n\\
&\leq t e^{2t M_t} \left(\epsilon_n+M_t^2 e^{\frac{t}{n}M_t}\frac{t}{2n}\right).
 \end{align*}
 This bound vanishes as $n\rightarrow\infty$.
 \end{proof}

To prove Theorem \ref{thm:finite} we use the Euler-type approximation established in
Lemma \ref{lem:Euler_app}. We show that the action of $T_n$ on a positive operator gives
a sequence of bounded from below operators such that the negative bounds vanish as $n$ 
goes to $\infty$.
\medskip

{\sc Proof of Theorem \ref{thm:finite}:}
First, we look at the each term in the Euler approximation $T_n$ separately.
For any $t$ and $s$ the complete dissipativity property (iii)  of $\mathcal{L}(s)$, assumed
in the statement of the theorem, implies
\begin{align*}
0&\leq (\id+t\mathcal{L}(s))(A^*)(\id+t\mathcal{L}(s))(A)=(A^*+t\mathcal{L}(s)(A^*))(A+t\mathcal{L}(s)(A))\\
& =A^*A+tA^*\mathcal{L}(s)(A)+t\mathcal{L}(s)(A^*)A+t^2\mathcal{L}(s)(A^*)\mathcal{L}(s)(A)\\
& \leq A^*A+t\mathcal{L}(s)(A^*A)+t^2\mathcal{L}(s)(A^*)\mathcal{L}(s)(A).
\end{align*}
Since $(\mathcal{L}(s)(A))^*(\mathcal{L}(s)(A))\leq \|\mathcal{L}(s)\|^2\|A\|,$ one gets
\begin{align}\label{norm_1}
0&\leq (\id+t\mathcal{L}(s))(A^*A)+t^2\|\mathcal{L}(s)\|^2\|A\|^2\\
& \leq(\id+t\mathcal{L}(s))(A^*A)+t^2M_s^2\|A\|^2.
\end{align}

Let us apply the above inequality to the operator$B $, where $B^*B:=\|A\|^2-A^*A$. Note that $\|B^*B\|\leq\|A\|^2$,
so $\|B\|\leq\|A\|$.
\begin{align}\label{norm_2}
0&\leq (\id+t\mathcal{L}(s))(\|A\|^2-A^*A)+t^2M_s^2\|A\|^2\\
& =\|A\|^2-(\id+t\mathcal{L}(s))(A^*A)+t^2M_s^2\|A\|^2
\end{align}

{From} the (\ref{norm_1}) and (\ref{norm_2}) we obtain
\begin{equation}\label{2}
-t^2M_s^2\|A\|^2\leq (\id+t\mathcal{L}(s))(A^*A)\leq (1+t^2M_s^2)\|A\|^2
\end{equation}
and therefore:
\begin{equation*}
-(1+t^2M_s^2)\|A\|^2\leq (\id+t\mathcal{L}(s))(A^*A) \leq (1+t^2M_s^2)\|A\|^2.
\end{equation*}
So we get
\begin{equation}\label{3}
\|(\id+t\mathcal{L}(s))(A^*A)\|\leq (1+t^2M_s^2)\|A\|^2 \ .
\end{equation}

Now, in order to bound the approximation $T_n$ we first derive the following auxiliary estimate.
For any fixed $n \geq 1$ we have:
\begin{equation}\label{bound}
\prod_{k=n}^1(\id+s\mathcal{L}(ks))(A^*A)\geq -s^2\|A\|^2M_{ns}^2(1+\frac{1}{n-1})^{n-1}\sum_{k=0}^{n-1} D(s)^k,
\end{equation}
where the value of $s$ is chosen to be such that
\begin{equation}\label{assumption_s}
D(s):=1+s^2M_{ns}^2 < {(1+\frac{1}{n-1})^{n-1}}/{(1+\frac{1}{n-2})^{n-2}},
\end{equation}
with the convention that $(1+\frac{1}{n-1})^{n-1} =1$, for $n=1$.

We prove this claim by induction. The statement holds for $n=1$ by (\ref{norm_2}).
Now, assume that \eq{bound} holds for $n-1$. Then
\begin{equation*}
\prod_{k=n-1}^{1}(\id+s\mathcal{L}(ks))(A^*A) +s^2\|A\|^2M_{(n-1)s}^2(1+\frac{1}{n-2})^{n-2}\sum_{k=0}^{n-2} D(s)^k\geq 0
\end{equation*}
Since the left-hand side is a positive operator, we can write it as $B^*B$. Then,
\begin{align*}
\prod_{k=n}^1(\id+s\mathcal{L}(ks))(A^*A)&=(\id+s\mathcal{L}(ns))(B^*B)\\&-s^2\|A\|^2M_{(n-1)s}^2(1+\frac{1}{n-2})^{n-2}
\sum_{k=0}^{n-2} D(s)^k\\
&\geq -s^2M_{ns}^2\|B^*B\|-s^2\|A\|^2M_{ns}^2(1+\frac{1}{n-2})^{n-2}\sum_{k=0}^{n-2} D(s)^k\, .
\end{align*}
Here, we used (\ref{norm_2}) and the fact that $M_t$ is monotone increasing.
This gives the following upper bound for $\Vert B^*B\Vert$:
\begin{align*}
\|B^*B\|&\leq \prod_{k=n-1}^{1} \|(\id+s\mathcal{L}(ks))(A^*A)\|+s^2\|A\|^2M_{(n-1)s}^2(1+\frac{1}{n-2})^{n-2}
\sum_{k=0}^{n-2}D(s)^k\\
&\leq\prod_{k=n-1}^{1}(1+s^2M_{ks}^2)\|A\|^2+s^2\|A\|^2M_{ns}^2(1+\frac{1}{n-2})^{n-2}\sum_{k=0}^{n-2}D(s)^k\\
& \leq\prod_{k=n-1}^{1}(1+s^2M_{ns}^2)\|A\|^2+s^2\|A\|^2M_{ns}^2(1+\frac{1}{n-2})^{n-2}\sum_{k=0}^{n-2}D(s)^k\\
&= \|A\|^2D(s)^{n-1}+s^2\|A\|^2M_{ns}^2(1+\frac{1}{n-2})^{n-2}\sum_{k=0}^{n-2}D(s)^k
\end{align*}
Therefore we obtain
\begin{align*}
&\prod_{k=n-1}^1(1+s\mathcal{L}(ks))(A^*A)\\
& \geq-s^2M_{ns}^2\|A\|^2D(s)^{n-1}-s^2M_{ns}^2(s^2M_{ns}^2+1)(1+\frac{1}{n-2})^{n-2}\|A\|^2\sum_{k=0}^{n-2}D(s)^k\\
&\geq-s^2M_{ns}^2\|A\|^2(1+\frac{1}{n-1})^{n-1}D(s)^{n-1}-s^2M_{ns}^2
(1+\frac{1}{n-1})^{n-1}\|A\|^2\sum_{k=0}^{n-2}D(s)^k\\&\geq-s^2M_{ns}^2
(1+\frac{1}{n-1})^{n-1}\|A\|^2\sum_{k=0}^{n-1}D(s)^k \ ,
\end{align*}
where to pass to the second inequality we use our assumption on $s$ \eq{assumption_s}.
This completes the proof of the bound (\ref{bound}).

To finish the proof of the theorem we use Lemma \ref{lem:Euler_app} to approximate the propagator and put 
$s=\frac{t}{n}$ in the bound (\ref{bound}), which yields
\begin{equation}\label{last_eq}
\prod_{k=n}^1(1+\frac{t}{n}\mathcal{L}(\frac{kt}{n}))(A^*A)\geq -\frac{t^2}{n^2}\|A\|^2M_t^2(1+\frac{1}{n-1})^{n-1}
\sum_{k=0}^{n-1} D(\frac{t}{n})^k.
\end{equation}
Since $D(\frac{t}{n})^n=(1+\frac{t^2}{n^2}M_t^2)^n \rightarrow 1$  as $n\rightarrow \infty$,
we get the estimate $D(\frac{t}{n})^k\leq 2$ for $1\leq k \leq n$. The right hand side of (\ref{last_eq}) 
is bounded from below by $-\frac{t^2}{n^2}\|A\|^2 e \, M_t^2 2n$, which vanishes in the limit $n\rightarrow \infty$.

To show the complete positivity of $\gamma_{t,0}$ note that any generator
$\mathcal{L}_\Lambda(t)$ satisfying the assumptions of the theorem
can be considered as the generator for a dynamics on $\cA\otimes \cB(\Cx^n)$, for any
$n\geq 1$, which satisfies the same properties, and which generates $\gamma_{t,s}\otimes\id$
acting on  $\cA\otimes \cB(\Cx^n)$. By the arguments given above,
these maps are positive for all $n$. Hence, the $\gamma_{t,s}$ are completely positive.
\boxendproof

\section{Lieb-Robinson bound}\label{sec:lrbounds}

Our derivation of the Lieb-Robinson bounds for $\gamma^{\Lambda}_{t,s}$ is based on 
a generalization of the strategy \cite{nachtergaele:2006} for reversible dynamics, and 
on \cite{poulin:2010} for irreversible dynamics with time-independent generators.
This allows us to cover the case of irreversible dynamics with time-dependent generators.
\medskip

\noindent
{\sc Proof of Theorem \ref{thm:lrbounds}:}
Consider the function $f:[s,\infty)\rightarrow \cA$ defined by
\begin{equation*}
f(t)=\cK\gamma_{t,s}^\Lambda(B),
\end{equation*}
where $\cK\in\cB_X$ and $B\in\cA_Y$, as in the statement of the theorem.
For $X\subset\Lambda$, let $X^c=\Lambda\setminus X$ and
define $\cL_{X^c}$ and $\bar{\cL}_{X}$ by
\begin{eqnarray*}
\cL_{X^c}(t)&=&\sum_{Z, Z\cap X= \emptyset}\mathcal{L}_Z(t)\\
\bar{\cL}_X(t)&=&\cL_X(t)-\cL_{X^c}(t).
\end{eqnarray*}
Clearly, $[\cK,\cL_{X^c}(t)]=0$. Using this property, we easily derive
the following expression for the derivative of $f$:
\begin{align*}
f'(t)&=\cK\mathcal{L}(t)\gamma_{t,s}^\Lambda(B) \\
& = \mathcal{L}_{X^c}(t)\cK\gamma_{t,s}^\Lambda(B)+\cK\bar{\cL}_{X}(t)\gamma_{t,s}^\Lambda(B) \\
& = \mathcal{L}_{X^c}(t)f(t)+\mathcal{K}\bar{\cL}_{X}(t)\gamma_{t,s}^\Lambda(B)\ ,
\end{align*}
Let $\gamma^{X^c}_{t.s}$ be the cocycle generated by $\cL_{X^c}(t)$. Then,
using the expression for $f'(t)$ we find
\begin{equation*}
f(t)= \gamma^{X^c}_{t,s}f(s)+\int_s^t\gamma_{t,r}^{X^c}\cK\bar{\cL}_{X}(r)
\gamma_{r,s}^\Lambda(B)dr \ .
\end{equation*}
Since $\gamma^{X^c}_{t,s}$ is norm-contracting and
$\Vert \cK\Vert_{\rm cb}$ is an upper bound for the 
$\Vert \cK\Vert$ regarded as an operator on $\cA_\Lambda$, for all
$\Lambda$, we obtain
\begin{equation}\label{f_norm_eneq}
\|f(t)\|\leq \|f(s)\|+\|\cK\|_{\rm cb}\int_s^t \|\bar{\cL}_{X}(r)\gamma_{r,s}^\Lambda(B)\|dr.
\end{equation}
Let us define the quantity
\begin{equation*}
C_B(X,t):= \sup_{\mathcal{T}\in\cB_X} \frac{\|\mathcal{T}
\gamma_{t,s}^\Lambda(B)\|}{\|\mathcal{T}\|_{\rm cb}}.
\end{equation*}
Note that we use the norm $\|\mathcal{T}\|_{\rm cb}$,
because, as mentioned before and in contrast to the usual operator norm,
 it is independent of $\Lambda$.
Then, we have the following obvious estimate:
\begin{equation*}
C_B(X,s)\leq \|B\|\delta_Y(X),
\end{equation*}
where $\delta_Y(X)=0$ if $X\cap Y=\emptyset$ and $\delta_Y(X)=1$ otherwise.
{From} the definition of the space
$\cB_X$ we get that $\mathcal{T}(B)=0$, when $\mathcal{T}\in\cB_X$, since
$B$ has a support in 
$Y$ and $Y\cap X=\emptyset.$ \\
Therefore  (\ref{f_norm_eneq}) implies that
\begin{equation*}
C_B(X,t)\leq C_B(X,s)+\sum_{Z\cap X \neq \emptyset} \int_s^t\|\mathcal{L}_Z(s)\| C_B(Z,s)ds.
\end{equation*}
Iterating this inequality we find the estimate:
\begin{equation*}
C_B(X,t)\leq \|B\| \sum_{n=0}^{\infty} \frac{(t-s)^n}{n!} \, a_n \ ,
\end{equation*}
where:
\begin{equation*}
a_n\leq \|\Psi\|_{t,\mu}^nC_\mu^{n-1}\sum_{x\in X}\sum_{y\in Y}F_\mu (d(x,y)),
\end{equation*}
for $n\geq 1$ and $a_0=1$, (recall that $C_\mu$ is a constant, that appears in a definition of $F_\mu$).
The following bound immediately follows from this estimate:
\begin{equation*}
\|\cK\gamma_{t,s}^\Lambda(B)\|\leq \frac{\|\cK\|_{\rm cb}\|B\|}{C_\mu} e^{\|\Psi\|_{t,\mu}
C_\mu (t-s)}\sum_{x\in X}\sum_{y\in Y}F_\mu(d(x,y)).
\end{equation*}
Using definition of $F_\mu$, we can rewrite this bound as
\begin{equation*}
\|\cK\gamma_{t,s}^\Lambda(B)\|\leq \frac{\|\cK\|_{\rm cb}\|B\|}{C_\mu} \|F\|\min(|X|, |Y|)
e^{-\mu(d(X,Y)-\frac{\|\Psi\|_{t,\mu} C_\mu}{\mu} (t-s))}.
\end{equation*}
So the Lieb-Robinson velocity of the propagation for every $t\in\mathbf{R}$ is
\begin{equation*}
v_{t,\mu}:=\frac{\|\Psi\|_{t,\mu} C_\mu}{\mu}.
\end{equation*}
\boxendproof

Note that the bound above depends only on the smallest of the supports of the
two observables. Therefore, in a situation where it makes sense to consider the
limit of infinite systems, one can get a non-trivial bound when one of the observables
has finite support but the support of the other is of infinite size (e.g., say half the system).

We would also like to point out that with the argument given in \cite{nachtergaele:2009b},
size of the support $|X|$, can be replaced by a suitable measure of the surface
area of the support, which gives a better estimate for observables with large
supports.

\section{Existence of the thermodynamic limit}\label{sec:thermodynamiclimit}

Our proof of existence of the thermodynamic limit mimics the method given in 
the paper \cite{nachtergaele:2006}.

\noindent
{\sc Proof of Theorem \ref{thm:thermodynamiclimit}:}
Denote $\mathcal{L}_n=\mathcal{L}_{\Lambda_n}$ and $\gamma_{t,s}^{\Lambda_n}=\gamma_{t,s}^{(n)}$.
Let $n>m$, then $\Lambda_m\subset\Lambda_n$ since we have the exhausting sequence of subsets in $\Gamma$. 
We will prove that for every observable $A\in\cA_X$ the sequence 
$(\gamma_{t,s}^n(A))_{n\geq 1}$ is a Cauchy sequence.
In order to do that for any local observable $A\in\cA_X$ we consider the function
\begin{equation*}
f(t):=\gamma_{t,s}^{(n)}(A)-\gamma_{t,s}^{(m)}(A) \ .
\end{equation*}
Calculating the derivative, we obtain
\begin{align*}
f^\prime(t)&=\mathcal{L}_n\gamma_{t,s}^{(n)}(A)-\mathcal{L}_m\gamma_{t,s}^{(m)}(A)\\
&= \mathcal{L}_n(t)(\gamma_{t,s}^{(n)}(A)-\gamma_{t,s}^{(m)}(A))+(\mathcal{L}_n(t)-\mathcal{L}_m(t))\gamma_{t,s}^{(m)}(A)\\
&=\mathcal{L}_n(t)f(t)+(\mathcal{L}_n(t)-\mathcal{L}_m(t))\gamma_{t,s}^{(m)}(A).
\end{align*}
The solution to this differential equation is
\begin{align*}
f(t)=\int_s^t \gamma_{t,r}^{(n)}(\mathcal{L}_n(r)-\mathcal{L}_m(r))\gamma_{r,s}^{(m)}Adr.
\end{align*}
Since $\gamma_{t,r}$ is norm-contracting, from this formula we get the estimate:
\begin{align}\label{estim-V}
\|f(t)\| &\leq \int_s^t \|(\mathcal{L}_n(r)-\mathcal{L}_m(r))\gamma_{r,s}^{(m)}(A) \|dr  \nonumber \\
&\leq \sum_{z\in \Lambda_n\setminus\Lambda_m}\sum_{Z\ni z}\int_s^t \|\Psi_Z(r)(\gamma_{r,s}^{(m)}(A))\|dr.\nonumber
\end{align}
Using the Lieb-Robinson bound and the exponential decay condition (\ref{Decay_L}),
which we assumed holds uniformly in $\Lambda$,  we find that
\begin{align*}
\|f(t)\|&\leq\frac{\|A\|}{C_\mu}\int_s^t e^{\mu v_{r,\mu} r} 
\sum_{z\in \Lambda_n\setminus\Lambda_m}\sum_{Z\ni z}\|\Psi_Z(r)\|_{\rm cb}
\sum_{x\in X}\sum_{y\in Z}F_\mu(d(x,y))dr\\
&\leq \frac{\|A\|}{C_\mu}\int_s^t e^{\mu v_{r,\mu} r} \sum_{z\in \Lambda_n\setminus\Lambda_m}\sum_{x\in X}\sum_{y\in\Gamma}
\sum_{Z\ni z,y}\|\Psi_Z(r)\|_{\rm cb}F_{\mu}(d(x,y))dr\\
&\leq \frac{\|A\|}{C_\mu}\|\Psi\|_{t,\mu}\int_s^t e^{\mu v_{r,\mu} r}dr \sum_{z\in \Lambda_n\setminus\Lambda_m}\sum_{x\in X}
\sum_{y\in\Gamma}  F_{\mu}(d(x,y))F_{\mu}(d(y,z))\\
&\leq \frac{\|A\|}{C_\mu}C_\mu\|\Psi\|_{t,\mu}\int_s^t e^{\mu v_{r,\mu} r}dr \sum_{z\in \Lambda_n\setminus\Lambda_m}
\sum_{x\in X}F_{\mu}(d(x,z))\\
&\leq  \frac{\|A\|}{C_\mu}C_\mu\|\Psi\|_{t,\mu}\int_s^t e^{\mu v_{r,\mu} r}dr |X| \sup_{x\in X}\sum_{z\in
\Lambda_n\setminus\Lambda_m}F_\mu(d(x,z)).
\end{align*}
Since $F_\mu$ is exponentially decaying when the distance $d(x,z)$ is increasing, we note that for 
$n, m\rightarrow\infty$, the last sum is goes to zero. Thus
\begin{equation*}
\|(\gamma_{t,s}^{(n)}-\gamma_{t,s}^{(m)})(A)\|\rightarrow 0, \textit{ as } n,m\rightarrow\infty.
\end{equation*}
Therefore the sequence $\{\gamma_{t,s}^{(n)}(A)\}_{n=0}^\infty$ is Cauchy and hence convergent. Denote the 
limit, and its extension to $\cA_\Gamma$, as $\gamma_{t,s}^\Gamma$.

To show that $\gamma_{t,s}^\Gamma$ is strongly continuous we notice that for
$0\leq s \leq t, r\leq T$, and any $A\in\cA_\Gamma^{\rm loc}$, we have
\begin{equation*}
\|\gamma_{t,s}^\Gamma (A)-\gamma_{r,s}^\Gamma(A)\|\leq \|\gamma_{t,s}^\Gamma(A)-\gamma_{t,s}^{(n)}(A)\|+
\|\gamma_{t,s}^{(n)}(A)-\gamma_{r,s}^{(n)}(A)\|+\|\gamma_{r,s}^{(n)}(A)-\gamma_{r,s}^\Gamma(A)\|,
\end{equation*}
for any $n\in\bN$ such that $A\in\cA_{\Lambda_n}$.
The strong continuity then follows from the strong convergence of $\gamma_{t,s}^{(n)}$ to
$\gamma_{t.s}^\Gamma$, uniformly in $s\leq t\in [0,T]$, and the strong continuity of
$\gamma_{t,s}^{(n)}$ in $t$. The continuity of the extension of $\gamma^\Gamma_{t,s}$ to all
of $A\in\cA_\Gamma$ follows by the standard density argument. The argument for continuity in
the second variable, $s$, is similar.
\boxendproof

\subsection*{Acknowledgments}

We thank Chris King for raising the question of Lieb-Robinson bounds for
irreversible quantum dynamics, and the referee for useful remarks and suggestions,
in particular for reminding  us about the recent work \cite{schuch:2010} and
raising the questions that led to some of the comments following Theorem \ref{thm:lrbounds}.

This work was supported by the National Science Foundation
under grants DMS-0757581 and DMS-1009502,
and by the France-Berkeley Fund under project \# 201013308.
V.A.Z. is thankful to the Mathematical Department of UC Davis for warm hospitality
and for the support from this Fund. B.N. acknowledges the support and hospitality of the 
Erwin Schr\"odinger International Institute for Mathematical Physics, Vienna.

\bibliographystyle{hamsplain}

\begin{thebibliography}{99}

\bibitem{amour:2009}
L.~ Amour, P.~Levy-Bruhl, and J.~Nourrigat.
\newblock Dynamics and Lieb-Robinson estimates for lattices of interacting anharmonic oscillators.
\newblock {\em Colloq. Math.} 118, no. 2, 609--648, 2010.
\mbox{arXiv:0904.2717}.

\bibitem{bachmann:2011}
S.~Bachmann, S.~Michalakis, B.~Nachtergaele, and R.~Sims,
\emph{Automorphic Equivalence within Gapped Phases of Quantum Lattice Systems},
\mbox{arXiv:1102.0842}.

\bibitem{bratteli:1980}
O.~Bratteli and A.~Kishimoto,
\emph{Generation of semigroups and two-dimensional quantum lattice systems},
J. Funct. Anal. \textbf{35} (1980) 344-368.

\bibitem{bratteli:1987}
O.~Bratteli and D.~W. Robinson, \emph{Operator algebras and quantum statistical
  mechanics}, 2 ed., vol.~1, Springer Verlag, 1987.

\bibitem{bratteli:1997}
\bysame, \emph{Operator algebras and quantum statistical mechanics}, 2 ed.,
  vol.~2, Springer Verlag, 1997.

\bibitem{davies:1977}
E.B.~Davies,
\emph{Irreversible Dynamics of Infinite Fermion Systems},
Commun. Math. Phys. \textbf{55} (1977) 231--258.

\bibitem{effros:2000}
E.G.~Effros and Z.-J.~Ruan,
\emph{Operator Spaces},
Oxford University Press, 2000.

\bibitem{fannes:1978}
M.~Fannes and A.~Verbeure,
\emph{Global thermodynamical stability and correlation inequalities},
J. Math. Phys. \textbf{19} (1978) 558--560.

\bibitem{gorini:1976}
V.~Gorini, A.~Kossakowski, and E.C.G.~Sudarshan,
\emph{Completely positive dynamical semigroups of N-Ievel systems},
J. Math. Phys. \textbf{17} (1976) 821--825

\bibitem{hamza:2009}
E.~Hamza, S.~Michalakis, B.~Nachtergaele, and R.~Sims, 
\emph{Approximating the ground state of gapped quantum spin systems}, 
J. Math. Phys. \textbf{50} (2009), 095213, 
%\mbox{arXiv:0904.4642}.

\bibitem{hastings:2004}
M.~B. Hastings, \emph{{L}ieb-{S}chultz-{M}attis in higher dimensions},
Phys. Rev. B \textbf{69} (2004) 104431.

\bibitem{hastings:2004b}
\bysame,
\emph{Locality in Quantum and Markov Dynamics on Lattices and Networks},
Phys. Rev. Lett. \textbf{93} (2004) 140402.

\bibitem{hastings:2006}
M.~B. Hastings and T.~Koma, \emph{Spectral gap and exponential decay of
  correlations}, Commun. Math. Phys. \textbf{265} (2006), 781--804,
  \mbox{arxiv:math-ph/0507.4708}.

\bibitem{hunter:2001}
J.K.~Hunter and B.~Nachtergaele,
\emph{Applied Analysis},
World Scientific, Singapore, 2001.

\bibitem{lieb:1972}
E.H. Lieb and D.W. Robinson, \emph{The finite group velocity of quantum spin
  systems}, Commun. Math. Phys. \textbf{28} (1972), 251--257.

\bibitem{lindblad:1976}
G.~Lindblad,
\emph{On the Generators o f Quantum Dynamical Semigroups}
Commun. Math. Phys. \textbf{48} (1976) 119--130.

\bibitem{matsui:1993}
T.~Matsui,
\emph{Markov Semigroups on UHF Algebras},
Rev. Math. Phys., {\bf 5} (1993), 587--600.

\bibitem{nachtergaele:2006}
B.~Nachtergaele, Y.~Ogata, and R.~Sims, \emph{Propagation of correlations in
  quantum lattice systems}, J. Stat. Phys. \textbf{124} (2006), 1--13,
  \mbox{arXiv:math-ph/0603064}.

\bibitem{nachtergaele:2009a}
B.~Nachtergaele, H.~Raz, B.~Schlein, and R.~Sims, \emph{{L}ieb-{R}obinson
  bounds for harmonic and anharmonic lattice systems}, Commun. Math. Phys.
  \textbf{286} (2009), 1073--1098, \mbox{arXiv:0712.3820}.

\bibitem{nachtergaele:2010}
B.~Nachtergaele, B.~Schlein, R.~Sims, S.~Starr, and V.~Zagrebnov, \emph{On the
  existence of the dynamics for anharmonic quantum oscillator systems}, Rev.
  Math. Phys. \textbf{22} (2010), 207--231, \mbox{arXiv:0909.2249}.

\bibitem{nachtergaele:2006a}
B.~Nachtergaele and R.~Sims, \emph{{L}ieb-{R}obinson bounds and the exponential
  clustering theorem}, Commun. Math. Phys. \textbf{265} (2006), 119--130,
  \mbox{arXiv:math-ph/0506030}.

\bibitem{nachtergaele:2009b}
\bysame, \emph{Locality estimates for quantum spin systems}, New Trends in
  Mathematical Physics. Selected contributions of the XVth International
  Congress on Mathematical Physics (V.~Sidovaricius, ed.), Springer Verlag,
  2009, pp.~591--614,  \mbox{arXiv:0712.3318}.

\bibitem{nachtergaele:2010b}
\bysame,
\emph{Much Ado About Something: Why Lieb-Robinson bounds are useful},
IAMP News Bulletin, October 2010, pp 22-29,  \mbox{arXiv:1102.0835}.

\bibitem{paulsen:2002}
V.~Paulsen,
\emph{Completely Bounded Maps and Operator Algebras},
Cambridge University Press, 2002.

\bibitem{poulin:2010}
D.~Poulin.
\newblock Lieb-Robinson bound and locality for general Markovian quantum dynamics.
\newblock {\em Phys. Rev. Lett.} \textbf{104}, 190401, 2010.

\bibitem{robinson:1968}
D.W.~Robinson,
\emph{Statistical Mechanics of Quantum Spin Systems. II},
Commun. Math. Phys. \textbf{7} (1968) 337--348.

\bibitem{robinson:1976}
\bysame,
\emph{Properties of propagation of quantum spin systems},
Austr. Math. Soc. \textbf{19} (1976) 387--399.

\bibitem{ruelle:1969}
D.~Ruelle,
\emph{Statistical Mechanics},
Benjamin, Reading, MA, 1969.

\bibitem{schuch:2010}
N.~Schuch, S.K.~Harrison, T.J.~Osborne, and J.~Eisert,
\emph{Information propagation for interacting particle systems},
\mbox{arXiv:1010.4576}

\bibitem{spohn:1980}
H.~Spohn,
\emph{Kinetic equations from Hamiltonian dynamics: Markovian limits},
Rev. Mod. Phys. \textbf{53} (1980) 569--615

\bibitem{streater:1968}
R.F.~Streater,
\emph{On Certain Non-Relativistic Quantized Fields},
Commun. Math. Phys. \textbf{7} (1968) 93--98.

\bibitem{unnerstall:1990}
T.~Unnerstall,
\emph{The Dynamics of Infinite Open Quantum Systems},
Lett. Math. Phys. \textbf{20} (1990) 183--187.

\bibitem{verbeure:2011}
A.F.~Verbeure,
\emph{Many-Body Boson Systems. Half a Century Later},
Springer Verlag (London), 2011.

\end{thebibliography}
\providecommand{\bysame}{\leavevmode\hbox to3em{\hrulefill}\thinspace}

\end{document}